\documentclass{article}
\usepackage{dcolumn}
\usepackage{bm}
\usepackage{verbatim}       

\usepackage[dvips]{graphicx}
\usepackage{amsthm}
\usepackage{amssymb}
\usepackage{bm}
\usepackage{amstext}
\usepackage{psfrag}
\usepackage{amsmath}
\usepackage{amsfonts}
\newfont{\bb}{msbm10 at 12pt}

\newcommand{\p}{\partial}

\newcommand{\dd}{{\rm d}}
\newcommand{\bd}{\begin{definition}}                
\newcommand{\ed}{\end{definition}}                  
\newcommand{\bc}{\begin{corollary}}                 
\newcommand{\ec}{\end{corollary}}                   
\newcommand{\bl}{\begin{lemma}}                     
\newcommand{\el}{\end{lemma}}                       
\newcommand{\bp}{\begin{proposition}}            
\newcommand{\ep}{\end{proposition}}                
\newcommand{\bere}{\begin{remark}}                  
\newcommand{\ere}{\end{remark}}                     

\newcommand{\bt}{\begin{theorem}}
\newcommand{\et}{\end{theorem}}

\newcommand{\be}{\begin{equation}}
\newcommand{\ee}{\end{equation}}

\newcommand{\bit}{\begin{itemize}}
\newcommand{\eit}{\end{itemize}}
\newtheorem{theorem}{Theorem}[section]
\newtheorem{corollary}[theorem]{Corollary}

\newtheorem{lemma}[theorem]{Lemma}
\newtheorem{proposition}[theorem]{Proposition}
\theoremstyle{definition}
\newtheorem{definition}[theorem]{Definition}
\theoremstyle{remark}
\newtheorem{remark}[theorem]{Remark}
\newtheorem{example}[theorem]{Example}

\begin{document}
%

\title{Non-imprisonment conditions on spacetime}

\author{E. Minguzzi \footnote{Dipartimento di Matematica Applicata, Universit\`a degli Studi di Firenze,  Via
S. Marta 3,  I-50139 Firenze, Italy. E-mail:
ettore.minguzzi@unifi.it}}

\date{}
\maketitle

\begin{abstract}
\noindent The non-imprisonment conditions on spacetimes are studied.
It is proved that the non-partial imprisonment property implies the
distinction property. Moreover, it is proved that feeble
distinction, a property which stays between weak distinction and
causality, implies non-total imprisonment. As a result the
non-imprisonment conditions can be included in the causal ladder of
spacetimes. Finally, totally imprisoned causal curves are  studied
in detail, and results concerning the existence and properties of
minimal invariant sets are obtained.
\end{abstract}

%


\section{Introduction}
Given a spacetime, i.e. a time oriented Lorentzian manifold, the
non-total future imprisonment condition is satisfied if no
future-inextendible causal curve can enter and remain in a compact
set. Analogously, the non-partial future imprisonment condition
requires that no future-inextendible causal curve reenters a compact
set indefinitely in the future. Since the formulation of these
properties uses only the causal structure of spacetime, it is
expected that they could be related with other conformal invariant
properties such as strong causality or distinction \cite{beem96}.

The classic book by Hawking and Ellis \cite[Sect. 6.4]{hawking73}
devotes to this issue several interesting propositions.
Unfortunately, some of them prove to be too weak and as today  the
relationship between non-partial (total) imprisonment  and the other
causality properties has yet to be clarified. The aim of this work
is to show that the non-imprisonment conditions can be included in
the causal ladder of spacetimes \cite{hawking74,minguzzi06c}, a
hierarchy of conformal invariant properties which is very useful in
order to establish at first sight the relationship between the most
common conformal invariant properties that have appeared in the
literature.

Note that in general the generic conformal invariant property does
not find a place in the causal ladder. For instance, the condition
of reflectivity \cite{kronheimer67,minguzzi06c}, defined by the
property $I^{+}(q) \subset I^{+}(p) \Leftrightarrow I^{-}(p) \subset
I^{-}(q)$, despite being conformal invariant, is not present by
itself in the ladder (one has to add to it some distinguishing
condition so as to obtain the level of causal continuity
\cite{hawking74}). Thus the fact that the non-partial and non-total
imprisonment conditions can find a place in the causal ladder is a
rather fortunate circumstance given the importance of these
conditions for the study of the causal structure of spacetime.

Figure \ref{impr} displays the relevant part of the causal ladder
and summarizes the final picture as will be given by this work. In
it I included some results recently obtained in \cite{minguzzi07e}.

I denote with $(M,g)$ a $C^{r}$ spacetime (connected, time-oriented
Lorentzian manifold), $r\in \{3, \dots, \infty\}$ of arbitrary
dimension $n\geq 2$ and signature $(-,+,\dots,+)$. The subset symbol
$\subset$ is reflexive, $X \subset X$. Several versions of the limit
curve theorem will be repeatedly used. The reader is referred to
\cite{minguzzi07c} for a sufficiently strong formulation.

\begin{figure}
\centering
\includegraphics[height=187pt]{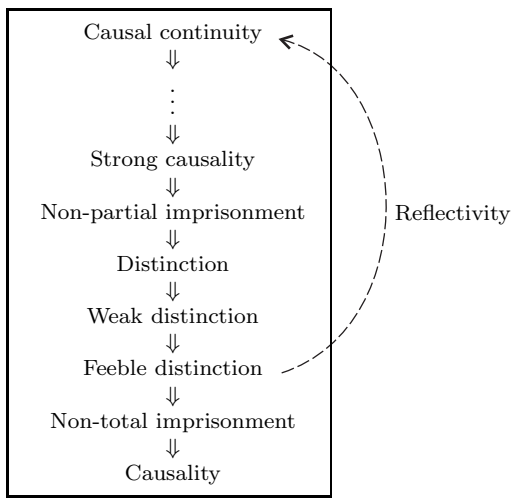}
\caption{The portion of the causal ladder which summarizes the
results of this work and \cite{minguzzi07e} concerning the causal
ladder. An arrow $A \Rightarrow B$ means that $A$ implies $B$ and
there are examples which show that $A$ differs from $B$.
Note that feeble distinction and reflectivity imply causal
continuity, for a proof of this result see \cite{minguzzi07e}.}
\label{impr}
\end{figure}


\section{Non-partial imprisonment}

In this section I shall introduce some definitions and basic
results. I will then consider the property of non-partial
imprisonment leaving the property of non-total imprisonment to the
next section.

\begin{definition}
A  future inextendible causal curve $\gamma: I \to M$, is {\em
totally future imprisoned} in the compact  set $C$ if there is $t
\in I$, such that for every $t'>t$, $t' \in I$, $\gamma(t') \in C$,
i.e. if it enters and remains in $C$. It is {\em partially future
imprisoned} if for every $t \in I$, there is $t'>t$, $t' \in I$,
such that $\gamma(t') \in C$, i.e. if it does not  remain in the
compact set it continually returns to it.  The curve {\em escapes to
infinity in the future} if it is not partially future imprisoned in
any compact set. Analogous definitions hold in the past case.
\end{definition}

\begin{remark}
In Hawking and Ellis \cite[p. 194]{hawking73} it is stated that a
future inextendible causal curve $\gamma: [0,b) \to M$ which is not
partially imprisoned in a compact set, intersects every compact set
only a finite number of times. However, this statement is incorrect
unless very strong differentiability conditions are imposed on the
curve. The point is that the causal curve may escape and reenter the
same compact set infinitely often while the parameter does not go to
$b$. Consider for instance 1+1 Minkwoski spacetime $\dd s^2=-\dd
t^2+\dd x^2$ and the $C^{k-2}$, $k>2$, timelike curve
$\gamma=(t,x(t))$ with $x(t)=0$ for $t\le 0$,
$x(t)=\frac{1}{k+2}(\tanh t)^k \sin (1/\tanh t)$ for $ t > 0 $. It
enters and escapes the compact set $ [-1,1] \times [-1,0]$ an
infinite number of times on any neighborhood of $t=0$.
\end{remark}

The previous definitions extend to the spacetime

\begin{definition}
A spacetime is {\em  non-total future imprisoning} if no future
inextendible causal curve is totally future imprisoned in a compact
set. A spacetime is {\em  non-partial future imprisoning} if no
future inextendible causal curve is partially future imprisoned in a
compact set. Analogous definitions hold in the past  case.
\end{definition}

Actually, Beem  proved \cite[theorem 4]{beem76} that a spacetime is
non-total future imprisoning if and only if it is non-total past
imprisoning, thus in the non-total case one can simply speak of the
{\em non-total imprisoning} property (condition $N$, in Beem's
terminology \cite{beem76}). This result will also be obtained in the
next section. A spacetime which is both non-partial future
imprisoning and non-partial past imprisoning is simply said to be
{\em non-partial imprisoning}.

The following result is immediate

\begin{proposition}
The non-total  imprisonment condition implies causality.
\end{proposition}

\begin{proof}
Assume causality does not hold. A closed causal curve can be made
inextendible while keeping the same compact image, simply by
extending the parametrization so as to make many rounds over the
original curve. The result is a inextendible curve whose image is
contained in a compact set (since it is a compact set itself), in
contradiction with the non-total imprisonment condition.
\end{proof}

Carter constructed an  example of causal but total imprisoning
spacetime. This classical example can be found in figure 39 of
\cite{hawking73}. Thus causality and non-total imprisonment do not
coincide.

The next result is well known \cite[Prop. 6.4.7]{hawking73},
\cite[Prop. 3.13]{beem96}.

\begin{proposition}
If the spacetime is strongly causal then it is non-partial
imprisoning.
\end{proposition}
Strong causality differs from non-partial imprisonment, see for
instance figure 38 of \cite{hawking73}.

The next observation is important for the placement of the
non-partial imprisonment  condition  in the causal ladder

\begin{proposition} \label{nix}
If a  spacetime  is non-partial future (resp. past) imprisoning then
it  is  past (resp. future) distinguishing. In particular
non-partial imprisoning spacetimes are distinguishing.
\end{proposition}

\begin{proof}
I give  the proof in the ``non partial future imprisoning - past
distinguishing'' case, the other case being analogous.

 Assume $(M,g)$ is non-partial future imprisoning. If
$(M,g)$ is non-past distinguishing there are $x \ne z$ such that
$I^{-}(x)=I^{-}(z)$. Let $U\ni x$ be a relatively compact set such
that $z \notin \bar{U}$, and let $V\ni z$ relatively compact set
such that $x \notin \bar{V}$, $\bar{U}\cap\bar{V}=\emptyset$. Take
$x_1 \in I^{-}(x)\cap \dot{U}$, then there is a timelike curve
$\sigma_1^z$ which connects $x_1$ to $z$. Let $z_1 \in \sigma_1^z
\cap \dot{V} \subset I^{-}(z)$, and parametrize $\sigma_1^z$ so that
$x_1=\sigma^z_1(0)$ and $z_1=\sigma^z_1(1)$. There is a timelike
curve $\sigma_1^x$ which connects $z_1$ to $x$. Let $x_2 \in
\sigma_1^x \cap \dot{U}\subset I^{-}(x)$, and parametrize
$\sigma_1^x$ so that $z_1=\sigma^x_1(1)$ and  $x_2=\sigma^x_1(2)$.
Continue in this way and obtain sequences $x_n \in \dot{U}$, $z_n
\in \dot{V}$, $\sigma_n^x$, $\sigma_n^z$. The timelike curve
\[
\sigma=
 \ldots\circ \sigma_2^x\vert_{[3,4]}  \circ \sigma_2^z\vert_{[2,3]} \circ \sigma_1^x\vert_{[1,2]} \circ \sigma_1^z\vert_{[0,1]}\]
 is future inextendible and is partially future imprisoned in both
$\bar{U}$ and $\bar{V}$. The contradiction proves that $(M,g)$ is
past distinguishing.
\end{proof}

\begin{figure}
\centering
\includegraphics[width=7.5cm]{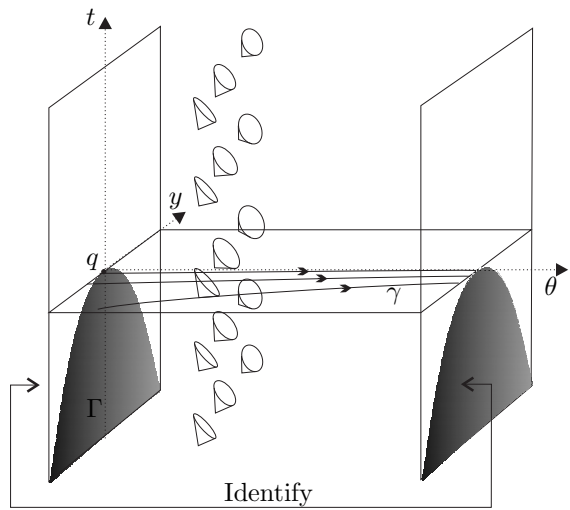}
\caption{A partial future imprisoning spacetime which is past
distinguishing. The set $\Gamma$ has been removed. It is  made of
all the points which can be connected to $q$ with a causal curve
that intersects $\theta=0$ at most at the endpoints.} \label{nonp}
\end{figure}

The next example proves that past  distinction differs from
non-partial future imprisonment (analogously future distinction
differs from non-partial past imprisonment).
\begin{example} \label{pdv}
Consider the spacetime $N=\mathbb{R} \times S^1 \times \mathbb{R}$
of coordinates $(t,\theta,y)$, $\theta \in [0,2\pi)$, and metric
\[
g=-\dd t\otimes    \dd \theta-   \dd \theta\otimes \dd t+t^2(\dd y-y
\dd \theta)^2+(\dd y+y \dd \theta)^2
\]
The vector field $\p/\p \theta$ is Killing and the field $\p/\p
\theta-y\p/\p y $ is lightlike on the null surface $t=0$, while
$\p/\p t$ is lightlike and future directed everywhere. If $V$ is the
tangent vector to a future  directed  causal curve then $g(V,\p/\p
t)\le 0$ which reads $\dd \theta [V]\ge 0$. The equality sign holds
only if $V \propto \p/\p t$, while  in all the other cases $\dd
\theta [V]> 0$. But it is also $g(V,V)\le 0$ which leads to $\dd
t[V] \dd \theta[V] \ge 0$, from which it follows that $t$ is a
quasi-time function, that is, it is non-decreasing over every causal
curve. The curve $\gamma=(t(\lambda), \theta(\lambda),y(\lambda))$
with $t(\lambda)=0$, $\theta(\lambda)=\lambda$, $y=-\exp(-\lambda)$
is a lightlike line partially imprisoned in the compact set
$[-1,1]\times [\pi/2,\pi]\times [-1,1]$. Call $\Gamma$ the set of
events on the surface $\theta=0$ which can be  connected to
$q=(0,0,0)$ through a causal curve which intersects the surface
$\theta=0$ only at the endpoints. Note that no point of $\gamma$ can
be connected with a causal curve to $q$, indeed the causal curve
would have to be lightlike in order to keep itself in the surface
$t=0$, which would imply the coincidence with $\gamma$ which,
however, does not pass through $q$. Thus $\gamma \cap
\Gamma=\emptyset$.

Remove $\Gamma$ from the spacetime, then $\gamma$ is still partially
imprisoned in the new spacetime $(M,g\vert_M)$ (see figure
\ref{nonp}) but while $(N,g)$ was not past distinguishing,
$(M,g\vert_M)$ is in fact past distinguishing. Indeed, the only set
of points at which past distinction can be violated is the $\theta$
axis, however, thanks to the removal of $\Gamma$ the points on it
have all distinct chronological pasts. Indeed, if $r_1<r_2$,
$r_1,r_2$ belong to the $\theta$ axis, then it can't be $r_2 \in
\overline{I^{-}(r_1)}$, for there would be a sequence of timelike
curves starting from a neighborhood of $r_2$ and reaching $r_1$
which is impossible since they would intersect $\Gamma$.
\end{example}

If from the spacetime of example \ref{pdv} one also removes a set
analogous to $\Gamma$, that is the set of events on the surface
$\theta=0$ which can be reached from $q$ with a causal curve which
intersects the surface $\theta=0$ only at the endpoints, then one
gets a spacetime that is distinguishing but not non-partial
imprisoning, thus distinction and non-partial imprisonment differ.

\section{Non-total imprisonment}

\begin{definition}
Let $\gamma: \mathbb{R} \to M$ be a causal curve. Denote with
$\Omega_f(\gamma)$ and $\Omega_p(\gamma)$ the following sets
\begin{eqnarray}
\Omega_f(\gamma)&= \bigcap_{t \in \mathbb{R}}
\overline{\gamma_{[t,+\infty)}} \\
%
%
\Omega_p(\gamma)&= \bigcap_{t \in \mathbb{R}}
\overline{\gamma_{(-\infty,t]}}
\end{eqnarray}
\end{definition}
They give, respectively, the set of accumulation points in the
future  of $\gamma$ and the set of accumulation points in the past
of $\gamma$. The sets $\Omega_f$ and $\Omega_p$ are well known from
the study of dynamical systems \cite[Sect. 3.2]{perko91}. The points
of $\Omega_f(\gamma)$ are called $\omega$-limit points of $\gamma$,
while the points of $\Omega_p(\gamma)$ are called $\alpha$-limit
points of $\gamma$. Note, however, that the analogy with dynamical
systems is not complete because so far no flow has been defined on
$M$.

\begin{proposition} \label{kic}
Let $\gamma: \mathbb{R} \to M$ be a  inextendible causal curve. The
set $\Omega_f(\gamma)$ is closed. The curve $\gamma$ is partially
future imprisoned in a compact set iff $\Omega_f(\gamma) \ne
\emptyset$. The curve  $\gamma$ is  totally future imprisoned in a
compact  set $C$ iff $\Omega_f(\gamma) \ne \emptyset$ and
$\Omega_f(\gamma)$ is compact. In this  case $\Omega_f(\gamma)$ is
the intersection of all the compact sets in which $\gamma$ is
totally future imprisoned, moreover, $\Omega_f(\gamma)$ is
connected. Analogous statements hold in the past case. Finally, for
every causal curve $\gamma$, the closure of its image is given by
$\overline{\gamma}=\Omega_p(\gamma)\cup \gamma \cup
\Omega_p(\gamma)$.
\end{proposition}

\begin{proof}
The closure is a consequence of the definition as  intersection of
closed sets.


If $\gamma$ is  partially future imprisoned in the compact set $C$
then for every $n \in \mathbb{N}$ there is $t_n \in \mathbb{R}$ such
that $x_n=\gamma(t_n) \in C$ and $t_n>n$. If $x \in C$ is an
accumulation point for $x_n$, there is a subsequence $x_{n_k}=
\gamma(t_{n_k})$ such that $x_{n_k} \to x$. Choose $t \in
\mathbb{R}$ then every neighborhood $U\ni x$ contains $x_{n_k}$ for
large $k$, thus $x \in \overline{\gamma_{[t,+\infty)}}$ and since
$t$ is arbitrary $x \in \Omega_f(\gamma)$.
For the converse, assume $\Omega_f(\gamma)$ is not empty, and take
$x \in \Omega_f(\gamma)$. Let $V\ni x$ be a neighborhood of compact
closure, then $\gamma$ is partially future imprisoned in the compact
set $C=\overline{V}$.

Assume $\gamma$ is totally future imprisoned in a compact set $C$.
Since $T$ can be chosen so large that $\gamma_{[T,+\infty)} \subset
C$, it is $\Omega_f(\gamma)\subset C$, and in particular
$\Omega_f(\gamma)$ is compact. Call $I$ the intersection of all the
compact sets totally future imprisoning $\gamma$, then since
$\Omega_f(\gamma)\subset C$ holds for any such compact set $C$,
$\Omega_f(\gamma)\subset I$.

Now, assume only that $\Omega_f(\gamma)$ is non-empty and compact.
Let $h$ be an auxiliary complete Riemannian metric on $M$, and let
$\rho$ be the corresponding continuous distance function. By the
Hopf-Rinow (Heine-Borel) theorem any closed and bounded set of
$(M,h)$ is compact. Thus $\Gamma_\epsilon=\{y \in M:
\rho(y,\Omega_f(\gamma)) \le \epsilon \}$ is compact and
$\bigcap_{\epsilon >0} \Gamma_\epsilon = \Omega_f(\gamma)$. But
$\gamma$ is totally future imprisoned in each $\Gamma_\epsilon$,
$\epsilon>0$. Indeed, if not $\gamma$ intersects indefinitely the
set $S_{\epsilon/2}=\{y \in M:
\rho(y,\Omega_f(\gamma))=\epsilon/2\}$, which is compact and thus
there would be a accumulation point $z \in S_{\epsilon/2} \cap
\Omega_f(\gamma)$ a contradiction since $S_{\epsilon/2} \cap
\Omega_f(\gamma)=\emptyset$. Thus $\gamma$ is totally future
imprisoned in a compact set iff $\Omega_f(\gamma)$ is non-empty and
compact. From $\bigcap_{\epsilon
>0} \Gamma_\epsilon = \Omega_f(\gamma)$, it follows that $I\subset
\Omega_f(\gamma)$, and using the other inclusion,
$I=\Omega_f(\gamma)$.

As for the connectedness of $\Omega_f(\gamma)$, without loss of
generality we can assume $\gamma$ entirely contained in the compact
set $C$, and we already know that $\Omega_f(\gamma) \subset C$. If
there are two disjoint non-empty closed sets $A$ and $B$ such that
$\Omega_f=A\cup B$, then there are two  open sets $A'\supset A$,
$B'\supset B$, such that $\overline{A'}\cap
\overline{B'}=\emptyset$.  Since, by definition of $\Omega_f$,
$\gamma$ is partially imprisoned in $\overline{A'}$ and
$\overline{B'}$ it crosses infinitely often both sets and there is a
sequence of points $z_r=\gamma(t_r) \in \gamma \subset C$, $t_r \to
+\infty$, $z_r \notin \overline{B'}\cup \overline{A'}$. Thus there
is an accumulation point $z \in C\backslash(A \cup B)$ a
contradiction since $z \in \Omega_f(\gamma)$. The proof in the past
case is analogous.

For the last statement let $\gamma$ have domain $(a,b)$ (finiteness
irrelevant), and let $x \in \overline{\gamma}$. For some sequence
$t_n \in (a,b)$, $\gamma(t_n) \to x$. Either $t_n$ admits a
subsequence which converges to $t_0 \in (a,b)$, in which case by
continuity $x=\gamma(t_0)\in \gamma$, or there is a subsequence
which converges to $b$, in which case $x \in \Omega_f(\gamma)$ or
finally, there is a subseqeunce which congerges to $a$, in which
case $x \in \Omega_p(\gamma)$.

\end{proof}

\begin{definition}
A lightlike line is an achronal inextendible causal curve.
\end{definition}
Since every causal curve with endpoints which is not a lightlike
geodesic can be varied to give a timelike curve joining the same
points \cite[Prop. 4.5.10]{hawking73}, a lightlike line is
necessarily a lightlike geodesic which, again by achronality,
maximizes the Lorentzian distance between any of its points.
Conversely, a inextendible lightlike geodesic which maximizes the
Lorentzian distance between any of its points is clearly a line
since the Lorentzian length calculated along the curve vanishes
(indeed this is the definition given by \cite[Prop. 8.12]{beem96}).

\begin{proposition} \label{bai} Let $\mathcal{C}$ be the chronology violating set of $(M,g)$ and let $\gamma: \mathbb{R} \to M$ be a
causal curve
\begin{itemize}
\item[(i)] Let $\alpha$ be a lightlike line such that $\alpha\subset
\Omega_f(\gamma)$ then $\Omega_f(\alpha)\cup \Omega_p(\alpha)
\subset \Omega_f(\gamma)$.
\item[(ii)] If $\gamma \cap \mathcal{C}=\emptyset$ then $\Omega_f(\gamma)$ is achronal (and thus $\Omega_f(\gamma)\cap \mathcal{C}=\emptyset$). Moreover, given $y \in \Omega_f(\gamma)$  there passes through $y$ one and only one lightlike line $\alpha$.
This line is such that $\alpha\subset \Omega_f(\gamma)$.
\item[(iii)] If $\gamma \cap \mathcal{C}=\emptyset$ and $\gamma$ is a lightlike line then
either $\gamma\subset \Omega_f(\gamma)$ or $\gamma \cap
\Omega_f(\gamma)=\emptyset$. More generally,\footnote{I am indebted
to an anonymous referee for this statement and its proof.} if
$\gamma \cap \mathcal{C}=\emptyset$, but $\gamma$ is  only a causal
curve then $\gamma$  cannot leave $\Omega_f(\gamma)$ if it ever
enters it.

\end{itemize}
Analogous statements hold in the past case.
\end{proposition}

\begin{proof}
Let $x \in \Omega_f(\alpha)\cup \Omega_p(\alpha)$ and take $T \in
\mathbb{R}$, and $U\ni x$. There is $t'$ such that $\alpha(t') \in
U$. But $\alpha(t') \in \Omega_f(\gamma)$ and $U$ is a neighborhood
also for $\alpha(t')$ thus there is $t>T$, such that $\gamma(t) \in
U$. Since $T$ and $U$ are arbitrary, $x \in \Omega_f(\gamma)$.

Assume $\Omega_f(\gamma)$ is not achronal then there are $x_1, x_2
\in \Omega_f(\gamma)$ such that $(x_1,x_2) \in I^+$. Using the fact
that $I^+$ is open there are neighborhoods $U_1 \ni x_1$ and $U_2
\ni x_2$ such that $U_1\times U_2 \subset I^{+}$. Since $x_1, x_2
\in \Omega_f(\gamma)$ it is possible to find $t_2<t_1$ such that
$\gamma(t_2) \in U_2$ and $\gamma(t_1) \in U_1$, thus  $\gamma(t_2)
\le \gamma(t_1) \ll \gamma(t_2)$ and hence $\gamma \cap
\mathcal{C}\ne \emptyset$, a contradiction.

Take $y \in \Omega_f(\gamma)$, and assume without loss of generality
that $\gamma$ is parametrized with respect to $h$-length where $h$
is a complete Riemannian metric. It is possible to find a sequence
$t_k$, $t_{k+1}>t_k+k$, such that $\gamma(t_k) \to y$. By the limit
curve theorem case (2) \cite{minguzzi07c} the segments
$\gamma_k=\gamma \vert_{[t_k,t_{k+1}]}$ have both endpoints that
converge to $y$ and since their $h$-length goes to infinity and
$\gamma_k \cap \mathcal{C}=\emptyset$, there is a lightlike line
$\alpha$ passing through $y$ which is a limit (cluster) curve for
$\gamma_k$.

There can't be another  lightlike line $\alpha'$ passing through $y$
indeed it is possible to show that it cannot be distance maximizing.
Indeed, let $w_1,w_2 \in \alpha$, $w_1<y<w_2$, and analogously let
$w'_1,w'_2 \in \alpha'$, $w'_1<y<w'_2$. The limit curve theorem
which has allowed to construct $\alpha$ states also that $(w_2,w_1)
\in \bar{J}^{+}$. Now, $\alpha'$ and $\alpha$ have different tanget
vectors at $y$, otherwise, being geodesics they would coincide, thus
rounding off the corners at $y$ it follows $w'_1 \ll w_2$ and $w_1
\ll w'_2$ which together with $(w_2,w_1) \in \bar{J}^{+}$ give,
since $I^{+}$ is open, $w'_1 \ll w'_2$, thus $\alpha'$ is not a
line, a contradiction.

Now, $\alpha\subset \Omega_f(\gamma)$, indeed, take $w \in \alpha$,
and let $U\ni w$ and $T \in \mathbb{R}$. Since $\alpha$ is a limit
curve, $U$ intersects all but a finite number of $\gamma_k$, in
particular it is possible to find $s \in \mathbb{N}$ such that
$t_s>T$ and $\gamma_s$ intersects $U$. Thus there is $t'_s \in
[t_s,t_{s+1}]$ such that $\gamma(t'_s) \in U$ and $t'_s \ge t_s>T$.
From the arbitrariness of $U$ and $T$, $w \in \Omega_f(\gamma)$.

Finally, assume that $z \in \gamma \cap \Omega_f(\gamma)$ and that
$\gamma$ is a lightlike line. By point (ii) through $z$ there passes
a unique lightlike line $\alpha$, and moreover such line has the
property $\alpha\subset\Omega_f(\gamma)$. But through $z$ there
passes already $\gamma$ thus $\alpha=\gamma$. More generally, if
$\gamma$ is only a causal curve, let $z \in \gamma \cap
\Omega_f(\gamma)$, $z=\gamma(t)$.  Split $\gamma$ in the two curves
$\gamma_p=\gamma\vert_{(-\infty, t]}$ and $\gamma_f=\gamma\vert_{[t,
+\infty)}$ and do the same with the lightlike line $\alpha$ passing
through $z$. The causal curve $\gamma_f\circ \alpha_p$ is achronal
because if not there are points $z_p \in \alpha_p$, $z_f\in
\gamma_f$, with $z_p \ll z_f$. But since $z_p \in
\Omega_f(\gamma)=\Omega_f(\gamma_f)$, it would be $\gamma_f\cap
\mathcal{C}\ne\emptyset$ a contradiction. Thus since $\gamma_f\circ
\alpha_p$ is achronal it is a lightlike line and hence it coincides
with $\alpha$, in particular $\gamma_f$ is  contained in
$\Omega_f(\gamma)$.

The proofs in the past case are analogous.
\end{proof}

\begin{corollary}
Non-total past imprisonment is equivalent to non-total future
imprisonment.
\end{corollary}

\begin{proof}
Assume non-total past imprisonment holds, then the spacetime is
causal. If non-total future imprisonment does not hold then there is
a future imprisoned causal curve $\eta$, and $\Omega_f(\eta)$ is
compact and non-empty (Prop. \ref{kic}). Take $p\in \Omega_f(\eta)$,
since the spacetime is chronological, there is a lightlike line
passing through $p$ (Prop. \ref{bai}(ii)) contained and hence
totally past imprisoned in $\Omega_f(\eta)$. The contradiction
proves non-total past imprisonment implies non-total future
imprisonment, the other direction being analogous.
\end{proof}

Propositions \ref{kic} and \ref{bai}  imply that if $\gamma \cap
\mathcal{C}=\emptyset$ and the inextendible causal curve $\gamma$ is
partially future imprisoned in a compact set
 then $\Omega_f(\gamma)$ is non-empty, generated by lightlike lines, and that there
is a privileged field of future directed null directions over
$\Omega_f(\gamma)$. Introduce a complete Riemannian metric on $M$,
and normalize the field of directions so as to obtain a field of
future directed lightlike vectors $n: \Omega_f(\gamma) \to
T\Omega_f(\gamma)$ which satisfies $\nabla_n n=h(x) n$, for some
scalar field $h$. It is then possible to define over
$\Omega_f(\gamma)$ the integral flow $\phi: \mathbb{R} \times
\Omega_f(\gamma) \to \Omega_f(\gamma) $ of $n$,
\[
\frac{\dd \phi_t(x_0)}{\dd t}=n(\phi_t(x_0)).
\]
which has the peculiarity of not having fixed points since $n \ne
0$. By proposition \ref{bai}(ii), $\Omega_f(\gamma)$ is invariant
under the flow $\phi_t$. Since the concept of invariance under the
flow $\phi$, is actually independent of the Riemannian metric chosen
to define $\phi$, it is convenient to state it in an equivalent but
clearer way as follows
\begin{definition} \label{pds}
 A non-empty closed subset $\Omega$ of $M$ is said to be {\em
invariant}, if (i) through each one of its points there passes one
and only one lightlike line and (ii) the entire line is contained in
$\Omega$. An invariant set is {\em minimal} if it has no invariant
proper subset.
\end{definition}

\begin{remark} \label{pda}
The definition of invariant set must embody the requirement of
non-emptiness, otherwise the empty set would be an invariant set and
no invariant set but the empty set would be minimal. Unfortunately,
in many references about dynamical systems, where similar
definitions are introduced, the non-emptiness condition is
incorrectly omitted
 (for instance \cite[p. 184]{hartman64}). Note that
according to the terminology of this work a minimal invariant set is
an invariant set and hence it is non-empty.
\end{remark}

\begin{lemma} \label{per}
The union of a finite family of invariant sets is an invariant set.
The intersection of an arbitrary family of invariant sets, if
non-empty, is an invariant set.
\end{lemma}

\begin{proof}
They are a trivial consequence of the definitions.
\end{proof}

There may be other invariant sets inside $\Omega_f(\gamma)$, indeed
for every $y \in \Omega_f(\gamma)$ there is a unique lightlike line
$\alpha \subset \Omega_f(\gamma)$ passing through $y$, and
$\Omega_f(\alpha) \cup \Omega_p(\alpha) \subset \Omega_f(\gamma)$.
Using again proposition \ref{bai}(ii), it follows that both
$\Omega_f(\alpha)$ and $\Omega_p(\alpha)$ satisfy conditions (i) and
(ii) of invariant sets (definition \ref{pds}). However, they could
be empty. They are certainly non-empty if $\gamma$ is not only
partially imprisoned but also totally imprisoned. Indeed, in this
case since $\alpha \subset\Omega_f(\gamma)$, and $\Omega_f(\gamma)$
is compact, $\alpha$ is totally imprisoned both in the past and in
the future and hence both $\Omega_p(\alpha)$ and $\Omega_f(\alpha)$
are non-empty and compact (proposition \ref{kic}).

Thus a totally future imprisoned curve leads to a partial order of
invariant sets (where the order is the usual inclusion) inside
$\Omega_f(\gamma)$. Actually, it is possible to prove that there
exists a minimal invariant set as the next proposition shows.
%
%
%

\begin{theorem} \label{bpo}
Let $\eta$ be a inextendible causal curve totally future imprisoned
in a compact set $C$, and let $\eta \cap \mathcal{C}=\emptyset$ with
$\mathcal{C}$ the chronology violating set of $(M,g)$. Then there is
a  minimal invariant set $\Omega \subset \Omega_f(\eta) \subset C$.
Through every point of $\Omega$ there passes one and only one
lightlike line, this lightlike line is entirely contained in
$\Omega$ and for every lightlike line $\alpha \subset \Omega$ it is
$\overline{\alpha}=\Omega_f(\alpha)=\Omega_p(\alpha)=\Omega$.
Another consequence is that all the points belonging to $\Omega$
share the same chronological past and future. An analgous version
holds with $\eta$ past imprisoned.
\end{theorem}

\begin{proof}
 Since $\eta$ is totally future
imprisoned, $\Omega_f(\eta)$ is non-empty and compact. Consider the
set  of all the
 invariant subsets of $\Omega_f(\eta)$ ordered by inclusion. This
set is non-empty since $\Omega_f(\eta)$ itself is invariant. By
Hausdorff's maximum principle (equivalent to Zorn's lemma and the
axiom of choice) there is a maximal chain  of invariant sets,
$\mathcal{Z}$, and the intersection of all the elements of the chain
gives a set $\Omega$ which, if non-empty, is an invariant set (lemma
\ref{per}) which has to be minimal otherwise the chain would not be
maximal. Since $\Omega_f(\eta)$ is compact all the elements of
$\mathcal{Z}$ are non-empty and compact, thus $\Omega$ is non-empty
being the intersection of a nested family of non-empty compact sets.
This last result is standard in topology (\cite[theorem
3.1.1]{engelking89}, in some references it is called Cantor's
intersection lemma), I include here the proof.


If $\bigcap_{U \in \mathcal{Z}} U=\emptyset$ then taking
complements, $\bigcup_{U \in \mathcal{Z}} U^C=M$, thus the sets
$U^C$ with $U \in \mathcal{Z}$ would give an open covering of $M$
and hence of $\Omega_f(\eta)$. Extract a finite subcovering of this
last set, and label the elements of the covering $U_1^C$, $U_2^C$,
$\ldots$, $U^C_n$ with $U_{j+1} \subset U_j$. Thus $U_n^C \supset
\Omega_f(\gamma)$ which is impossible since $U_n \subset
\Omega_f(\gamma)$ is non-empty. The contradiction proves that
$\Omega$ is non-empty and hence minimal invariant.

Since $\Omega$ is invariant through every point of it there passes a
unique lightlike line entirely contained in $\Omega$. Taken $\alpha
\subset \Omega$, it is $\Omega_f(\alpha),\Omega_p(\alpha) \subset
\Omega$, but since $\Omega$ is actually the minimal invariant set
the equality must hold. Since
$\overline{\alpha}=\Omega_p(\alpha)\cup\alpha\cup\Omega_f(\alpha)$,
we have also $\overline{\alpha}=\Omega$.

Finally, let $x,z \in \Omega$, $x \ne z$, and let $z' \in I^{+}(z)$.
Consider the lightlike line  $\gamma_x$ passing through $x$. Since
$\Omega_f(\gamma_x)=\Omega$, $z \in \Omega_f(\gamma_x)$, and in
particular since there is a whole neighborhood of $z$ in the past of
$z'$, $z' \in I^{+}(x)$, and hence $I^{+}(z) \subset I^{+}(x)$.
Exchanging the roles of $x$ and $z$, and  of pasts and futures, we
get $I^{+}(x)=I^{+}(z)$ and $I^{-}(x)=I^{-}(z)$.

\end{proof}

Recall that a feebly distingushing spacetime is a spacetime for
which no two causally related events can have the same chronological
past and chronological future. It is a weaker condition than weak
distinction in which the causal relation of the events is not
mentioned. Moreover, it is stronger than causality and differs from
both causality and weak distinction \cite{minguzzi07e}.

Actually, there is not very much difference between weak distinction
and feeble distinction, as it is quite difficult to produce examples
of spacetimes which are feebly distinguishing but non-weakly
distinguishing. Nevertheless,  this difference is mentioned and
enphasized here because the next theorem stated with feeble
distinction is slightly stronger than with weak distinction.

The  next result  improves Hawking and Ellis' \cite{hawking73} who
assume past or future distinction. Note that without this result the
relative strength of non-total imprisonment and weak distinction
would have been left open, a fact that has so far forbidden the
placement of non-total imprisonment into the causal ladder.

\begin{corollary}
If a spacetime is feebly distinguishing then it is non-total
imprisoning.
\end{corollary}

\begin{proof}
It follows trivially from theorem \ref{bpo}, because if $\eta$ is a
causal curve totally future imprisoned in a compact set, then there
is a minimal invariant set $\Omega \subset \Omega_f(\eta)$ (it is
easy to see that feeble distinction implies chronology and hence
$\eta \cap \mathcal{C}=\emptyset$). Let $\gamma \subset \Omega$ be a
lightlike line, then the points of $\gamma$ are causally related but
share the same chronological past and future in contradiction with
feeble distinction.
\end{proof}

The fact that feeble distinction differs from non-total imprisoning
follows, for instance, from figure 1(A) of \cite{minguzzi07e}.

\section{Conclusions}
The properties of non-total imprisonment and non-partial
imprisonment have been placed into the causal ladder. The placement
of non-partial imprisonment has required the proof that non-partial
imprisonment implies distinction  (Prop. \ref{nix}) and the
production of an example which shows that the two levels do indeed
differ (figure \ref{nonp}).

More interesting and rich has proved the study of total
imprisonment. To start with I considered a partial future imprisoned
causal curve $\gamma$. I proved that on every event of the limit set
$\Omega_f(\gamma)$  there passes one and only one lightlike line,
and that this line is contained in $\Omega_f(\gamma)$. This result
allows to define a flow on the same set and makes it possible to
study this situation in analogy with dynamical systems. If $\gamma$
is actually totally imprisoned then it is possible to infer the
existence of a minimal invariant set on which feeble distinction is
violated. Thus feeble distinction implies non-total imprisonment,
and the non-imprisonment conditions can thus be placed in the causal
ladder.

There remains some interesting work to be done along the lines
followed here in the study of totally imprisoned curves. For
instance, in 2+1 dimensions one could perhaps apply the
Poincar\'e-Bendixon-Schwartz theorem to the invariant set to infer
more results on the connection between imprisonment and manifold
topology. However, in order to follow this path it is necessary to
prove that the invariant set is actually smooth or has a  higher
degree of differentiability than Lipschitz continuity (which follows
from its achronality, see proposition \ref{bai}(ii)). Similar ideas
have indeed been applied in the study of chronology violating
spacetimes \cite{hawking92,tiglio98} but often without adressing the
differentiability issues which indeed arise because the
Poincar\'e-Bendixon-Schwartz theorem, as formulated by Schwartz,
would require the invariant set to be a $C^2$ manifold.


\section*{Acknowledgments}
I warmly thank a referee for very helpful comments, in particular
for stressing the importance of the non-emptinesss condition (see
remark \ref{pda}).
This work has been partially supported by GNFM of INDAM and
by MIUR under project PRIN 2005 from Universit\`a di Camerino.


\end{document}